\documentclass{amsart}
\copyrightinfo{2017}{Nicholas Coxon}
\usepackage{amssymb,url}
\usepackage[foot]{amsaddr}
\newtheorem{theorem}{Theorem}[section]
\newtheorem{proposition}[theorem]{Proposition}
\newtheorem{lemma}[theorem]{Lemma}

\theoremstyle{definition}

\theoremstyle{remark}

\numberwithin{equation}{section}
\usepackage[noend]{algpseudocode}
\usepackage{algorithm}

\newcommand{\ElseIf}[1]{\State\hspace{-1.5em}Else if\ #1\ \hspace*{-0.333em}:}

\newcommand{\floor}[1]{\left\lfloor #1\right\rfloor}
\newcommand{\ceil}[1]{\left\lceil #1\right\rceil}
\newcommand{\N}{\mathbb{N}}
\newcommand{\Z}{\mathbb{Z}}
\newcommand{\F}{\mathbb{F}}
\newcommand{\LeftFunction}{PrepareLeft}
\newcommand{\RightFunction}{PrepareRight}
\newcommand{\A}{A}
\newcommand{\M}{M}
\DeclareMathOperator{\mult}{\mathsf{M}}
\DeclareMathOperator{\bigO}{\mathcal{O}}
\makeatletter
\newcommand*{\bdiv}{%
	\nonscript\mskip-\medmuskip\mkern5mu%
	\mathbin{\operator@font div}\penalty900\mkern5mu%
	\nonscript\mskip-\medmuskip
}
\newcommand*{\bmodstar}{%
	\nonscript\mskip-\medmuskip\mkern5mu%
	\mathbin{\operator@font mod^\ast}\penalty900\mkern5mu%
	\nonscript\mskip-\medmuskip
}
\makeatother
\let\originalleft\left
\let\originalright\right
\renewcommand{\left}{\mathopen{}\mathclose\bgroup\originalleft}
\renewcommand{\right}{\aftergroup\egroup\originalright}
\begin{document}
%
%-----------------------------------------
% Title, authors, abstract, etc.
%-----------------------------------------
%
\title[Fast Hermite interpolation and evaluation over characteristic two]{Fast 
Hermite interpolation and evaluation over finite fields of characteristic two}
\author[Nicholas Coxon]{Nicholas Coxon}
\address{INRIA and Laboratoire d'Informatique de l'\'{E}cole polytechnique, Palaiseau, France.}
\email{nicholas.coxon@inria.fr}
\date{\today}
\thanks{This work was supported by Nokia in the framework of the common laboratory between Nokia Bell Labs and INRIA}
\begin{abstract} This paper presents new fast algorithms for Hermite 
interpolation and evaluation over finite fields of characteristic two. The 
algorithms reduce the Hermite problems to instances of the standard multipoint 
interpolation and evaluation problems, which are then solved by existing fast 
algorithms. The reductions are simple to implement and free of multiplications, 
allowing low overall multiplicative complexities to be obtained. The algorithms 
are suitable for use in encoding and decoding algorithms for multiplicity codes.
\end{abstract}
\maketitle
%
%-----------------------------------------
% Body
%-----------------------------------------

\section{Introduction}\label{sec:introduction}

Hermite interpolation is the problem of computing the coefficients of a polynomial given the values of its derivatives up to a given order at one or more evaluation points. The inverse problem, that of evaluating the derivatives of the polynomial when given its coefficients, is sometimes referred to as Hermite evaluation. Over fields of positive characteristic~$p$, the $i$th formal derivative vanishes identically for $i\geq p$. Consequently, it usual to consider Hermite interpolation and evaluation with respect to the Hasse derivative over fields of small positive characteristic.

For now, let $\F$ simply denote a field. Then, for $i\in\N$, the map $D^i:\F[x]\rightarrow\F[x]$ that sends $F\in\F[x]$ to the coefficient of $y^i$ in $F(x+y)\in\F[x][y]$ is called the $i$th Hasse derivative on $\F[x]$. For distinct evaluation points $\omega_0,\dotsc,\omega_{n-1}\in\F$ and positive integer multiplicities $\ell_0,\dotsc,\ell_{n-1}$, the Hermite interpolation problem over~$\F$ asks that we compute the coefficients of a polynomial $F\in\F[x]$ of degree less than $\ell=\ell_0+\dotsb+\ell_{n-1}$ when given $(D^iF)(\omega_j)$ for $j\in\{0,\dotsc,\ell_i-1\}$ and $i\in\{0,\dotsc,n-1\}$. The corresponding instance of the Hermite evaluation problem asks that we use the coefficients of $F$ to compute the $\ell$ derivatives of the interpolation problem. Different versions of the problems specify different bases on which the polynomials are required to be represented. In this paper, the problems are considered with respect to the monomial basis $\{1,x,x^2,\dotsc\}$ of $\F[x]$ only.

The boundary case $\ell_0=\dotsb=\ell_{n-1}=1$ 
corresponds to standard multipoint interpolation and evaluation, allowing both 
problems to be solved with $\bigO(\mult(\ell)\log \ell)$ operations in~$\F$ by 
the use of remainder trees and fast Chinese remainder 
algorithms~\cite{fiduccia1972,moenck1972,borodin1974,bostan2003,bostan2004,bernstein2004,hoeven2016}
 (see also~\cite[Chapter~10]{gathen2013}). Here, $\mult(\ell)$ denotes the 
number of operations required to multiply two polynomials in $\F[x]$ of degree 
less than $\ell$, which may be taken to be in $\bigO(\ell(\log 
\ell)\log\log\ell)$~\cite{schonhage1971,schonhage1976,cantor1991}. The 
complexity of solving the standard interpolation and evaluation problems 
reduces to $\bigO(\mult(\ell))$ operations when the evaluation points form a 
geometric progression~\cite{bostan2005}. Similarly, fast Fourier transform 
(FFT)~\cite{cooley1965} based interpolation and evaluation offer complexities 
as low as $\bigO(\ell\log\ell)$ operations on certain special sets of 
evaluation points.

For the opposing boundary case of $n=1$, the Hermite interpolation and evaluation problems both reduce to computing Taylor expansions. Indeed, it follows directly from the definition of Hasse derivatives that
\begin{equation}
	\label{eqn:hasse-taylor-expansion}
	F
	=\sum_{i\in\N}
	(D^iF)(\omega)
	\left(x-\omega\right)^i
	\quad
	\text{for $F\in\F[x]$ and $\omega\in\F$}.
\end{equation}
Consequently, Hermite interpolation and evaluation at a single evaluation point can be performed in $\bigO(\mult(\ell)\log\ell)$ operations in general~\cite{borodin1974,gathen1990,gathen1997}, $\bigO(\mult(\ell))$ operations if $(\ell-1)!$ is invertible in the field~\cite{aho1975,vari1974} (see also~\cite{gathen1997,bini1994}), and $\bigO(\ell\log\ell)$ operations if the field has characteristic equal to two~\cite{gao2010}.

The first quasi-linear time algorithms for solving the general Hermite problems were proposed by Chin~\cite{chin1976}. Truncating the Taylor expansion~\eqref{eqn:hasse-taylor-expansion} after $i$ terms gives the residue of degree less than $i$ of $F$ modulo $(x-\omega)^i$. Based on this observation, Chin's evaluation algorithm begins by using a remainder tree to compute the residues of the input polynomial modulo $(x-\omega_i)^{\ell_i}$ for $i\in\{0,\dotsc,n-1\}$. The Taylor expansion of each residue at its corresponding evaluation point is then computed to obtain the truncated Taylor expansion of the input polynomial. The interpolation problem can be solved by reversing these steps, with the residues combined by a fast Chinese remainder algorithm. It follows that the general Hermite interpolation and evaluation problems may be solved with $\bigO(\mult(\ell)\log\ell)$ operations~\cite{chin1976,olshevsky2000} (see also~\cite{bini1994,pan2001}).

In this paper, we present new algorithms for Hermite interpolation and 
evaluation over finite fields of characteristic two. The algorithms require the 
set of evaluation points to equal the field itself, and their corresponding 
multiplicities to be balanced, with $\left|\ell_i-\ell_j\right|\leq 1$ for 
$i\neq j$. While not solving the general interpolation and evaluation problems 
over these fields, the algorithms are suitable for use in multivariate Hermite 
interpolation and evaluation algorithms~\cite{coxon2017}, encoding and decoding 
algorithms for multiplicity codes~\cite{kopparty2014,coxon2017} and the codes 
of Wu~\cite{wu2015}, and private information retrieval protocols based on these 
codes~\cite{woodruff2007,augot2014}.

Over a characteristic two finite field of order $q$, the restricted problems may be solved with $\bigO(\mult(\ell)\log q+\ell\log\ell/q)$ operations by existing algorithms. The algorithms presented in this paper yield the same complexities, but benefit by their simplicity and the low number of multiplications they perform. When $\ell$ is a multiple of $q$, as occurs in some encoding and decoding contexts, the Hermite interpolation algorithm presented here performs $\ell/q$ standard interpolations over the $q$ evaluation points, followed by $\bigO(\ell\log\ell/q)$ additions. The Hermite evaluation algorithm performs $\bigO(\ell\log\ell/q)$ additions, followed by $\ell/q$ standard evaluations over the $q$ points. Using the generic bound of $\bigO(\mult(q)\log q)$ operations for solving the standard problems leads to the above bound on solving the Hermite problems in this special case.

We are not prohibited from using faster FFT-based interpolation and evaluation algorithms to solve the standard problems when they are supported by the field. Moreover, we have the option of using ``additive FFT'' algorithms~\cite{cantor1989,mateer2008,gao2010,bernstein2013,bernstein2014,lin2014,lin2016,lin2016a,chen2017}, which are specific to characteristic two finite fields and allow evaluation and interpolation over the $q$ points of the field to be performed with $\bigO(q\log^2q)$ or $\bigO(q(\log q)\log\log q)$ additions, depending on the degree of the field, and $\bigO(q\log q)$ multiplications. With these algorithms and $\ell$ a multiple of $q$, the Hermite interpolation and evaluation algorithms perform $\bigO(\ell(\log^2q+\log\ell/q))$ or $\bigO(\ell((\log q)\log\log q+\log\ell/q))$ additions, and only $\bigO(\ell\log q)$ multiplications.

When $\ell$ is not a multiple of $q$, the Hermite interpolation and evaluation algorithms still perform $\ceil{\ell/q}-1$ standard interpolations or evaluations over the $q$ evaluation points, but each must also solve one instance of a slightly generalised version of the corresponding standard problem. However, these more general problems reduce to the standard problems at the cost of $\bigO(\mult(q))$ operations for performing one division with remainder of polynomials of degree less than $q$ (see~\cite[Section~9.1]{gathen2013}). Consequently, the algorithms retain their simplicity in this case.

The reduction from Hermite to standard problems is provided in Section~\ref{sec:setup}, where we develop divide-and-conquer algorithms for solving the Hermite interpolation and evaluation problems when $\ell/q$ is a power of two. The problems for arbitrary $\ell$ can be reduced to this special case by zero padding. However, this approach almost doubles the size of the initial problem when $\ell/q$ is slightly larger than a power of two, leading to large jumps in complexity. Instead, in Sections~\ref{sec:evaluation} and~\ref{sec:interpolation} we address the problems for arbitrary $\ell$ by transferring across ideas from pruned and truncated FFT algorithms~\cite{markel1971,sorensen1993,hoeven2004,harvey2009,harvey2010,larrieu2017}, which are used to smooth similar unwanted jumps in the complexities of FFT-based evaluation and interpolation schemes. We are consequently able to solve the Hermite interpolation and evaluation problems with better complexity than obtained by zero padding.

\section{Properties of Hasse derivatives}\label{sec:hasse}

We begin by recalling some basic properties of Hasse derivatives.

\begin{lemma}\label{lem:hasse} Let $F,G\in\F[x]$, $\alpha,\beta,\omega\in\F$ and $i\in\N$. Then
\begin{enumerate}
	\item\label{hasse:linearity} $D^i(\alpha F+\beta G)=\alpha(D^iF)+\beta(D^iG)$,
	\item\label{hasse:taylor-coefficient} $(D^iF)(\omega)$ is equal to the coefficient of $x^i$ in $F(x+\omega)$,
	\item\label{hasse:multiplicity} $(D^jF)(\omega)=0$ for $j\in\{0,\dotsc,i-1\}$ if and only if $(x-\omega)^i$ divides $F$,
	\item\label{hasse:monomial} $D^ix^k=\binom{k}{i}x^{k-i}$ for $k\in\N$, and
	\item\label{hasse:composition} $D^i\circ D^j=\binom{i+j}{i}D^{i+j}$ for $j\in\N$.
\end{enumerate}
\end{lemma}

Properties~\eqref{hasse:linearity} and~\eqref{hasse:taylor-coefficient} of Lemma~\ref{lem:hasse} follow readily from the definition of Hasse derivatives provided in the introduction. Property~\eqref{hasse:multiplicity} follows from Property~\eqref{hasse:taylor-coefficient}.
Property~\eqref{hasse:monomial} follows from the definition of Hasse derivatives and the binomial theorem. Property~\eqref{hasse:composition} follows from Properties~\eqref{hasse:linearity} and~\eqref{hasse:monomial}, and the binomial identity
\begin{equation*}
	\binom{k-j}{i}\binom{k}{j}
	=\binom{i+j}{i}\binom{k}{i+j}
	\quad\text{for $i,j,k\in\N$}.
\end{equation*}

For $\ell>0$, let $\F[x]_\ell$ denote the space of polynomials in $\F[x]$ that have degree less than $\ell$. Then existence and uniqueness for the general Hermite interpolation problem is provided by the following lemma.

\begin{lemma}\label{lem:existence-uniqueness} Let $\omega_0,\dotsc,\omega_{n-1}\in\F$ be distinct, $\ell_0,\dotsc,\ell_{n-1}$ be positive integers, and $\ell=\ell_0+\dotsb+\ell_{n-1}$. Then given elements $h_{i,j}\in\F$ for $i\in\{0,\dotsc,\ell_j-1\}$ and $j\in\{0,\dotsc,n-1\}$, there exists a unique polynomial $F\in\F[x]_\ell$ such that $(D^iF)(\omega_j)=h_{i,j}$ for $i\in\{0,\dotsc,\ell_j-1\}$ and $j\in\{0,\dotsc,n-1\}$.
\end{lemma}

Lemma~\ref{lem:existence-uniqueness} follows from Property~\eqref{hasse:multiplicity} of Lemma~\ref{lem:hasse}, which implies that the kernel of the linear map from $\F[x]_\ell$ to $\F^\ell$ given by $F\mapsto((D^iF)(\omega_j))_{0\leq i<\ell_j,0\leq j<n}$ can only contain multiples of the degree $\ell$ polynomial $\prod^{n-1}_{j=0}(x-\omega_j)^{\ell_j}$, and must therefore be trivial.

\section{Strategy over finite fields of characteristic two}\label{sec:setup}

Hereafter, we assume that $\F$ is finite of characteristic two. Let $q$ denote the order of the field, and enumerate its elements as $\omega_0,\dotsc,\omega_{q-1}$. Define $i\bdiv{j}=\floor{i/j}$ and $i\bmod{j}=i-\floor{i/j}j$ for $i,j\in\Z$ such that $j$ is nonzero. Then the Hermite interpolation problem we consider in the remainder of the paper can be stated as follows: given $(h_0,\dotsc,h_{\ell-1})\in\F^\ell$, compute the vector $(f_0,\dotsc,f_{\ell-1})\in\F^\ell$ such that $F=\sum^{\ell-1}_{i=0}f_ix^i$ satisfies $(D^{i\bdiv{q}}F)(\omega_{i\bmod{q}})=h_i$ for $i\in\{0,\dotsc,\ell-1\}$. The Hermite evaluation problem we consider is the inverse problem, asking that we compute the vector $((D^{i\bdiv{q}}F)(\omega_{i\bmod{q}}))_{0\leq i<\ell}$ when given the coefficient vector of $F\in\F[x]_\ell$.  We call $\ell$ the length of an instance of either problem, and observe that if $\ell\leq q$, then the problems reduce to standard multipoint interpolation and evaluation with evaluation points $\omega_0,\dotsc,\omega_{\ell-1}$.

In this section, we introduce the main elements of our algorithms by temporarily limiting our attention to instances of length $2^nq$ for some $n\in\N$. The algorithms take on their simplest form in this case, with each applying a simple reduction from the length $2^nq$ problem to two problems of length $2^{n-1}q$. Proceeding recursively, both algorithms ultimately reduce to problems of length $q$, which are then solved by existing standard interpolation and evaluation algorithms. The reductions employed by the algorithms are provided by the following lemma.

\begin{lemma}\label{lem:recursion} Let $n\in\N$ be nonzero, 
$F_0,F_1\in\F[x]_{2^{n-1}q}$ and
\begin{equation}\label{eqn:taylor-expansion}
	F=F_1(x^q-x)^{2^{n-1}}+F_0.
\end{equation}
Then for $\omega\in\F$ and $i\in\{0,\dotsc,2^n-1\}$,
\begin{equation*}\label{eqn:hasse-recursion}
	\left(D^iF\right)(\omega)
	=\begin{cases}
		\left(D^iF_0\right)(\omega) & \text{if $i<2^{n-1}$},\\
		\big(D^{i-2^{n-1}}(F_1 + D^{2^{n-1}}F_0)\big)(\omega)
		& \text{otherwise}.
	\end{cases}
\end{equation*}
\end{lemma}
\begin{proof} Let $n\in\N$ be nonzero, $F_0,F_1\in\F[x]_{2^{n-1}q}$ and define $F$ by~\eqref{eqn:taylor-expansion}. Then
\begin{equation*}
	F(x+\omega)
	=F_1(x+\omega)x^{2^{n-1}q}+F_1(x+\omega)x^{2^{n-1}}+F_0(x+\omega)
\end{equation*}
for $\omega\in\F$. Consequently, as  $2^{n-1}q\geq 2^n$, Property~\eqref{hasse:taylor-coefficient} of Lemma~\ref{lem:hasse} implies that
\begin{equation*}
	\left(D^iF\right)(\omega)
	=\begin{cases}
		\left(D^iF_0\right)(\omega) & \text{if $i<2^{n-1}$},\\
		\big(D^{i-2^{n-1}}F_1\big)(\omega)+\left(D^iF_0\right)(\omega)
		& \text{otherwise},
	\end{cases}
\end{equation*}
for $\omega\in\F$ and $i\in\{0,\dotsc,2^n-1\}$. Therefore, linearity of Hasse derivatives implies that the lemma will follow if we can show that $D^i=D^{i-2^{n-1}}\circ D^{2^{n-1}}$ for $i\in\{2^{n-1},\dotsc,2^n-1\}$. To this end, we use Lucas' lemma~\cite[p.~230]{lucas1878} (see also~\cite{fine1947}), which states that
\begin{equation}\label{eqn:lucas}
	\binom{u}{v}
	\equiv\binom{u\bdiv{2^r}}{v\bdiv{2^r}}
	\binom{u\bmod{2^r}}{v\bmod{2^r}}\pmod{2}
	\quad\text{for $u,v,r\in\N$}.
\end{equation}
By combining Lucas' lemma with Property~\eqref{hasse:composition} of Lemma~\ref{lem:hasse}, we find that
\begin{equation*}
	D^{i-2^{n-1}}\circ D^{2^{n-1}}
	=\binom{2^{n-1}+(i-2^{n-1})}{2^{n-1}}D^i
	=\binom{1}{1}\binom{i-2^{n-1}}{0}D^i
	=D^i
\end{equation*}
for $i\in\{2^{n-1},\dotsc,2^n-1\}$.
\end{proof}

Given a vector $(h_0,\dotsc,h_{2^nq-1})\in\F^{2^nq}$ that defines an instance of the Hermite interpolation problem, our algorithm recursively computes the corresponding polynomial $F\in\F[x]_{2^nq}$ as follows. If $n=0$, then we are in the base case of the recursion, and $F$ is recovered by a standard interpolation algorithm. If $n\geq 1$, then the algorithm is recursively called on $(h_0,\dotsc,h_{2^{n-1}q-1})$ and $(h_{2^{n-1}q},\dotsc,h_{2^nq-1})$. Lemmas~\ref{lem:existence-uniqueness} and~\ref{lem:recursion} imply that the recursive calls return $F_0$ and $F_1+D^{2^{n-1}}F_0$, where $F_0$ and $F_1$ are the unique polynomials in $\F[x]_{2^{n-1}q}$ that satisfy~\eqref{eqn:taylor-expansion}. Thus, the algorithm next recovers $F_1$ by computing $D^{2^{n-1}}F_0$ and adding it to $F_1+D^{2^{n-1}}F_0$. Finally, $F$ is computed by expanding~\eqref{eqn:taylor-expansion} as
\begin{equation}\label{eqn:expanded-taylor}
	F=F_1x^{2^{n-1}q}+F_1x^{2^{n-1}}+F_0.
\end{equation}
Given a polynomial $F\in\F[x]_{2^nq}$, the evaluation algorithm uses a standard evaluation algorithm in its base case of $n=0$, and if $n\geq 1$, then it simply reverses the steps of the interpolation algorithm by first computing $F_0$ and $F_1$, then $F_1+D^{2^{n-1}}F_0$, and finally recursively evaluating $F_0$ and $F_1+D^{2^{n-1}}F_0$. In both algorithms, the following lemma is used to compute derivatives.

\begin{lemma}\label{lem:derivative} Let $n\in\N$ be nonzero and $F=\sum^{2^{n-1}q-1}_{i=0}f_ix^i\in\F[x]$. Then
\begin{equation}\label{eqn:derivative}
	D^{2^{n-1}}F
	=\sum^{q/2-1}_{i=0}x^{2^ni}
	\sum^{2^{n-1}-1}_{j=0}f_{2^{n-1}(2i+1)+j}x^j.
\end{equation}
\end{lemma}
\begin{proof} Let $n\in\N$ be nonzero and $F=\sum^{2^{n-1}q-1}_{i=0}f_ix^i\in\F[x]$. Then Property~\eqref{hasse:monomial} of Lemma~\ref{lem:hasse} and Lucas' lemma, in the form of~\eqref{eqn:lucas}, imply that
\begin{equation*}
	D^{2^{n-1}}x^{2^{n-1}(2i+b)+j}
	=\binom{2i+b}{1}\binom{j}{0}x^{2^{n-1}(2i+b-1)+j}
	=bx^{2^{n-1}(2i+b-1)+j}
\end{equation*}
for $b\in\{0,1\}$, $i\in\N$ and $j\in\{0,\dotsc,2^{n-1}-1\}$. Therefore, writing $F$ in the form
\begin{equation*}
	F=
	\sum^1_{b=0}
	\sum^{q/2-1}_{i=0}
	\sum^{2^{n-1}-1}_{j=0}f_{2^{n-1}(2i+b)+j}x^{2^{n-1}(2i+b)+j}
\end{equation*}
and applying $D^{2^{n-1}}$ to each of its terms yields~\eqref{eqn:derivative}.
\end{proof}

\section{Evaluation algorithm}\label{sec:evaluation}

To solve the Hermite evaluation problem for arbitrary lengths we reduce to the special case of the preceding section by padding the input vector with zeros. Following the approach of pruned and truncated FFT algorithms, we lessen the penalty incurred by having to solve the larger problems by pruning those steps of the algorithm that are specific to the computation of unwanted entries in the output. Thus, we consider the following revised problem in this section: given the coefficients of a polynomial $F\in\F[x]_{2^nq}$ and $c\in\{1,\dotsc,2^nq\}$, compute $(D^{i\bdiv{q}}F)(\omega_{i\bmod{q}})$ for $i\in\{0,\dotsc,c-1\}$. The length $\ell$ Hermite evaluation problem is then captured by taking $n=\ceil{\log_2\ceil{\ell/q}}$ and $c=\ell$.

The Hermite evaluation algorithm is described in Algorithm~\ref{alg:hermite-evaluation}. The algorithm operates on a vector $(a_0,\dotsc,a_{2^nq-1})\in\F^{2^nq}$ that initially contains the coefficients of a polynomial $F\in\F[x]_{2^nq}$, and overwrites $a_i$ with $(D^{i\bdiv{q}}F)(\omega_{i\bmod{q}})$ for $i$ less than the input value $c$. The remaining entries of the vector are either unchanged or set to intermediate values from the computation. If $c>2^{n-1}q$, then the algorithm follows the steps described in the preceding section for the length $2^nq$ problem, with the exception that the recursive call used to evaluate $F_1+D^{2^{n-1}}F_0$ only computes the $c-2^{n-1}q$ values required for the output. If $c\leq 2^{n-1}q$, then the output depends on $F_0$ only, so only $F_0$ is computed (by the function \textsf{\LeftFunction}) and recursively evaluated. Once again, the recursion terminates with $n=0$, which is handle by an algorithm \textsf{Evaluate} that satisfies the specifications of Algorithm~\ref{alg:evaluation}.

\begin{algorithm}[h]
	\caption{$\textsf{Evaluate}((a_0,\dotsc,a_{q-1}),c)$}
	\label{alg:evaluation}
	\begin{algorithmic}[1]
		\Require $(a_0,\dotsc,a_{q-1})\in\F^q$ and $c\in\{1,\dotsc,q\}$ such that $\sum^{q-1}_{i=0}a_ix^i=F$ for some $F\in\F[x]_q$.
		\Ensure $a_i=F(\omega_i)$ for $i\in\{0,\dotsc,c-1\}$.
	\end{algorithmic}
\end{algorithm}

Algorithm~\ref{alg:evaluation} may be realised with a complexity of 
$\bigO(\mult(q)+\mult(c)\log c)$ operations in $\F$ by the use of remainder 
trees. For small values of $c$, one can apply Horner's rule for each of the $c$ 
evaluation points. Naive matrix-vector products are efficient for small $q$, 
while additive and (standard) multiplicative FFT algorithms become more 
efficient for large $q$. For multiplicative FFT algorithms to be used, the 
multiplicative group of the field must have smooth cardinality, and it is 
necessary to first reduce modulo $x^{q-1}-1$ at the cost of one addition. 
Moreover, if one has control over the enumeration of the field, then it is 
possible to obtain a better complexity for $c<q$ by using the truncated FFT 
algorithm of Larrieu~\cite{larrieu2017}.

\begin{algorithm}[h]
	\caption{$\textsf{HermiteEvaluate}((a_0,\dotsc,a_{2^nq-1}),c)$}
	\label{alg:hermite-evaluation}
	\begin{algorithmic}[1]
		\Require $(a_0,\dotsc,a_{2^nq-1})\in\F^{2^nq}$ and $c\in\{1,\dotsc,2^nq\}$ such that $n\in\N$ and $\sum^{2^nq-1}_{i=0}a_ix^i=F$ for some $F\in\F[x]_{2^nq}$.
		\Ensure $a_i=(D^{i\bdiv{q}}F)(\omega_{i\bmod{q}})$ for $i\in\{0,\dotsc,c-1\}$.
		\If{$n=0$}
			\State\label{eval:base}\Call{Evaluate}{$(a_0,\dotsc,a_{q-1}),c$}
			\hfill\textit{/* Algorithm~\ref{alg:evaluation} */}
		\ElseIf{$c>2^{n-1}q$}
			\For{$i=2^{n-1}(q+1)-1,2^{n-1}(q+1)-2,\dotsc,2^{n-1}$}\label{eval:taylor-start}
				\State $a_i\leftarrow a_i+a_{2^{n-1}(q-1)+i}$
			\EndFor\label{eval:taylor-end}
			\For{$i=q/2,\dotsc,q-1$}\label{eval:hasse-start}
				\For{$j=0,\dotsc,2^{n-1}-1$}
					\State $a_{2^ni+j}\leftarrow a_{2^ni+j}+a_{2^ni+j-(q-1)2^{n-1}}$
				\EndFor
			\EndFor\label{eval:hasse-end}
			\State\label{eval:large-c-left-recursion} \Call{HermiteEvaluate}{$(a_0,\dotsc,a_{2^{n-1}q-1}),2^{n-1}q$}
			\State\label{eval:large-c-right-recursion} \Call{HermiteEvaluate}{$(a_{2^{n-1}q},\dotsc,a_{2^nq-1}),c-2^{n-1}q$}	
		\Else
			\State\label{eval:compute-F_0} \Call{\LeftFunction}{$a,0$}
			\State\label{eval:small-c-left-recursion} \Call{HermiteEvaluate}{$(a_0,\dotsc,a_{2^{n-1}q-1}),c$}
		\EndIf
		\vspace{-0.65\baselineskip}
		\Statex\hspace*{-\algorithmicindent}\hrulefill
		\Function{\LeftFunction}{$(a_0,\dotsc,a_{2^nq-1}),c$}:
			\For{$i=\max(c,2^{n-1}),\dotsc,2^{n-1}q-1$}
				\State $a_i\leftarrow a_i+a_{2^{n-1}(q-1)+i}$
			\EndFor
			\For{$i=\max(c,2^{n-1}),\dotsc,2^n-1$}
				\State $a_i\leftarrow a_i+a_{2^n(q-1)+i}$
			\EndFor
		\EndFunction
	\end{algorithmic}
\end{algorithm}

\begin{proposition} Algorithm~\ref{alg:hermite-evaluation} is correct if Algorithm~\ref{alg:evaluation} is correctly implemented.
\end{proposition}
\begin{proof} Under the assumption that Algorithm~\ref{alg:evaluation} has been correctly implemented, we use induction to show that for all $n\in\N$, Algorithm~\ref{alg:hermite-evaluation} produces the correct output when given inputs $(a_0,\dotsc,a_{2^nq-1})\in\F^{2^nq}$ and $c\in\{1,\dotsc,2^nq\}$. Therefore, suppose that Algorithm~\ref{alg:evaluation} has been correctly implemented. Then for inputs with $n=0$, the algorithm trivially produces the correct output since Algorithm~\ref{alg:evaluation} is simply applied in this case. Let $n\geq 1$ and suppose that Algorithm~\ref{alg:hermite-evaluation} functions correctly for all inputs with smaller values of~$n$. Suppose that $(a_0,\dotsc,a_{2^nq-1})\in\F^{2^nq}$ and $c\in\{1,\dotsc,2^nq\}$ are given to the algorithm as inputs, and let $F\in\F[x]_{2^nq}$ be the corresponding polynomial for which the input requirements are satisfied. Let $F_0,F_1\in\F[x]_{2^{n-1}q}$ such that~\eqref{eqn:taylor-expansion} and, equivalently, \eqref{eqn:expanded-taylor} hold.
	
Suppose that $c>2^{n-1}q$. Then \eqref{eqn:expanded-taylor} implies that Lines~\ref{eval:taylor-start} and~\ref{eval:taylor-end} set $a_i$ equal to the coefficient of $x^i$ in $F_0$, and $a_{2^{n-1}q+i}$ equal to the coefficient of $x^i$ in $F_1$, for $i\in\{0,\dotsc,2^{n-1}q-1\}$. Consequently, Lemma~\ref{lem:derivative} implies that Lines~\ref{eval:hasse-start} to~\ref{eval:hasse-end} set $a_{2^{n-1}q+i}$ equal to the coefficient of $x^i$ in $F_1+D^{2^{n-1}}F_0$ for $i\in\{0,\dotsc,2^{n-1}q-1\}$. As these three lines do not modify $a_0,\dotsc,a_{2^{n-1}q-1}$, which contain the coefficients of~$F_0$, the induction hypothesis and Lemma~\ref{lem:recursion} imply that the recursive call of Lines~\ref{eval:large-c-left-recursion} sets
\begin{equation}\label{eqn:evaluation-left}
	a_i
	=\left(D^{i\bdiv{q}}F_0\right)
	\left(\omega_{i\bmod{q}}\right)
	=\left(D^{i\bdiv{q}}F\right)
	\left(\omega_{i\bmod{q}}\right)
\end{equation}
for $i\in\{0,\dotsc,2^{n-1}q-1\}$. Similarly, the recursive call of Line~\ref{eval:large-c-right-recursion} sets
\begin{equation}\label{eqn:evaluation-right}
	\begin{split}
	a_{2^{n-1}q+i}
	&=\left(D^{i\bdiv{q}}\left(F_1+D^{2^{n-1}}F_0\right)\right)
	\left(\omega_{i\bmod{q}}\right)\\
	&=\left(D^{2^{n-1}+(i\bdiv{q})}F\right)
	\left(\omega_{i\bmod{q}}\right)\\
	&=\left(D^{(2^{n-1}q+i)\bdiv{q}}F\right)
	\left(\omega_{(2^{n-1}q+i)\bmod{q}}\right)
	\end{split}
\end{equation}
for $i\in\{0,\dotsc,c-2^{n-1}q-1\}$. The algorithm stops at this point, and thus produces the correct output.

Suppose now that $c\leq 2^{n-1}q$. Then \eqref{eqn:expanded-taylor} implies that Line~\ref{eval:compute-F_0} sets $a_i$ equal to the coefficient of $x^i$ in $F_0$ for $i\in\{0,\dotsc,2^{n-1}q-1\}$. Consequently, the induction hypothesis and Lemma~\ref{lem:recursion} imply that \eqref{eqn:evaluation-left} holds for $i\in\{0,\dotsc,c-1\}$ after the recursive call of Line~\ref{eval:small-c-left-recursion} has been performed. The algorithm stops at this point, and thus produces the correct output.
\end{proof}

For $i,j\in\Z$ such that $j$ is nonzero, define $i\bmodstar{j}=i-(\ceil{i/j}-1)j$. The following proposition bounds the additive and multiplicative complexities of Algorithm~\ref{alg:hermite-evaluation} in term of those of Algorithm~\ref{alg:evaluation}.

\begin{proposition}\label{prop:evaluation-complexity} For $n\in\N$, define $\A_n,\M_n:\{1,\dotsc,2^nq\}\rightarrow\N$ as follows: $\A_n(c)$ and $\M_n(c)$ are respectively the number of additions and multiplications in~$\F$ performed by Algorithm~\ref{alg:hermite-evaluation} (for some implementation of Algorithm~\ref{alg:evaluation}) when given inputs $(a_0,\dotsc,a_{2^nq-1})\in\F^{2^nq}$ and $c\in\{1,\dotsc,2^nq\}$. Then
\begin{multline*}
	\A_n(c)
	\leq
	\A_0(q)
	\left(\ceil{c/q}-1\right)
	+\A_0(c\bmodstar{q})\\
	+\left(
		\frac{3}{4}\ceil{\log_2\ceil{c/q}}
		-\frac{1}{4}
	\right)
	\left(\ceil{c/q}-1\right)q
	+(2^n-1)q
\end{multline*}
and $\M_n(c)=\M_0(q)(\ceil{c/q}-1)+\M_0(c\bmodstar{q})$ for $n\in\N$ and $c\in\{1,\dotsc,2^nq\}$.
\end{proposition}
\begin{proof} For nonzero $n\in\N$, define the indicator function $\delta_n:\{1,\dotsc,2^nq\}\rightarrow\{0,1\}$ by $\delta_n(c)=1$ if and only if $c>2^{n-1}q$. Then given inputs $(a_0,\dotsc,a_{2^nq-1})\in\F^{2^nq}$ and $c\in\{1,\dotsc,2^nq\}$ for some nonzero $n\in\N$, Lines~\ref{eval:taylor-start} to~\ref{eval:hasse-end} of Algorithm~\ref{alg:hermite-evaluation} perform $\delta_n(c)(3/4)2^nq$ additions, Line~\ref{eval:large-c-left-recursion} performs $\delta_n(c)\A_{n-1}(2^{n-1}q)$ additions and $\delta_n(c)\M_{n-1}(2^{n-1}q)$ multiplications, and Line~\ref{eval:large-c-right-recursion} performs $\delta_n(c)\A_{n-1}(c-\delta_n(c)2^{n-1}q)$ additions and $\delta_n(c)\M_{n-1}(c-\delta_n(c)2^{n-1}q)$ multiplications. Furthermore, Line~\ref{eval:compute-F_0} performs $(1-\delta_n(c))2^{n-1}q$ additions, and Line~\ref{eval:small-c-left-recursion} performs $(1-\delta_n(c))\A_{n-1}(c-\delta_n(c)2^{n-1}q)$ additions and $(1-\delta_n(c))\M_{n-1}(c-\delta_n(c)2^{n-1}q)$ multiplications. Summing these contributions, it follows that
\begin{equation}
	\label{eqn:evaluation-A_n-formula-1}
	\A_n(c)
	=
	\A_{n-1}\left(c-\delta_n(c)2^{n-1}q\right)
	+2^{n-1}q
	+\delta_n(c)
	\left(\A_{n-1}\left(2^{n-1}q\right)+2^{n-2}q\right)
\end{equation}
and
\begin{equation}
	\label{eqn:evaluation-M_n-formula-1}
	\M_n(c)
	=
	\M_{n-1}\left(c-\delta_n(c)2^{n-1}q\right)
	+\delta_n(c)\M_{n-1}\left(2^{n-1}q\right)
\end{equation}
for nonzero $n\in\N$ and $c\in\{1,\dotsc,2^nq\}$. In particular,
\begin{equation}
	\label{eqn:full-length-recurrence}
	\A_n(2^nq)=2\A_{n-1}\left(2^{n-1}q\right)+\frac{3}{4}2^nq
	\quad\text{and}\quad
	\M_n(2^nq)=2\M_{n-1}\left(2^{n-1}q\right)
\end{equation}
for nonzero $n\in\N$. Thus,
\begin{equation}
	\label{eqn:full-length-complexity}
	\A_n(2^nq)=2^n\left(\A_0(q)+\frac{3}{4}nq\right)
	\quad\text{and}\quad
	\M_n(2^nq)=2^n\M_0(q)
	\quad\text{for $n\in\N$}.
\end{equation}
Substituting these equations into~\eqref{eqn:evaluation-A_n-formula-1} and~\eqref{eqn:evaluation-M_n-formula-1}, it follows that
\begin{equation}
	\label{eqn:evaluation-A_n-formula-2}
	\A_n(c)
	=\A_{n-1}\left(c-\delta_n(c)2^{n-1}q\right)
	+2^{n-1}q
	+\delta_n(c)2^{n-1}
	\left(\A_0(q)+\frac{3}{4}(n-1)q+\frac{q}{2}\right)
\end{equation}
and
\begin{equation}
	\label{eqn:evaluation-M_n-formula-2}
	\M_n(c)=\M_{n-1}\left(c-\delta_n(c)2^{n-1}q\right)+\delta_n(c)2^{n-1}\M_0(q)
\end{equation}
for nonzero $n\in\N$ and $c\in\{1,\dotsc,2^nq\}$.

For $n,c,\delta\in\N$ such that $n$ is nonzero, we have $\ceil{(c-\delta 2^{n-1}q)/q}=\ceil{c/q}-\delta2^{n-1}$ and $c-\delta 2^{n-1}q\bmodstar{q}=c\bmodstar{q}$. Consequently, the formula for $\M_n(c)$ stated in the proposition follows from~\eqref{eqn:evaluation-M_n-formula-2} by induction on $n$. If $n\in\N$ is nonzero, $c\in\{1,\dotsc,2^nq\}$ and
\begin{equation}\label{eqn:binary-expansion}
	\ceil{c/q}-1=i_0+i_1\cdot 2+\dotsb+i_{n-1}\cdot 2^{n-1},
\end{equation}
with $i_0,\dotsc,i_{n-1}\in\{0,1\}$, then $i_{n-1}=\delta_n(c)$. Therefore, it follows from \eqref{eqn:evaluation-A_n-formula-2} by induction on $n$, that
\begin{equation}\label{eqn:evaluation-A_n}
	\A_n(c)
	=\A_0\left(c\bmodstar{q}\right)
	+\left(2^n-1\right)q
	+\left(\A_0(q)+\frac{q}{2}\right)\left(\ceil{c/q}-1\right)
	+\frac{3}{4}q\sum^{n-1}_{k=0}2^ki_kk
\end{equation}
for $n\in\N$ and $c\in\{1,\dotsc,2^nq\}$, where $i_0,\dotsc,i_{n-1}\in\{0,1\}$ are the coefficients of the binary expansion~\eqref{eqn:binary-expansion}. Here,
\begin{equation*}
	\sum^{n-1}_{k=0}2^ki_kk
	\leq\max\left(\ceil{\log_2\ceil{c/q}}-1,0\right)
	\sum^{n-1}_{k=0}2^ki_k
	=\left(\ceil{\log_2\ceil{c/q}}-1\right)
	\left(\ceil{c/q}-1\right),
\end{equation*}
since $i_k=0$ if $k\geq\ceil{\log_2\ceil{c/q}}$. Combining this inequality with~\eqref{eqn:evaluation-A_n} yields the upper bound on $\A_n(c)$ stated in the proposition.
\end{proof}

The functions $\A_0$ and $\M_0$ defined in Proposition~\ref{prop:evaluation-complexity} describe the additive and multiplicative complexities of the implementation of Algorithm~\ref{alg:evaluation}. When $n=0$ or $c>2^{n-1}q$ for some nonzero $n$, as may be assumed when solving the Hermite evaluation problem, the third and fourth terms of the bound on $\A_n(c)$ are in $\bigO(c\log\ceil{c/q})$. By taking $\A_0$ and $\M_0$ to be in $\bigO(\mult(q)+\mult(c)\log c)$, and making the common assumption (used, for instance, in~\cite{gathen2013}) that $\mult(\ell)/\ell$ is a nondecreasing function of $\ell$, it follows that the length $\ell$ Hermite evaluation problem can be solved with $\bigO(\mult(\ell)\log q+\ell\log\ceil{\ell/q})$ operations in $\F$ by Algorithm~\ref{alg:hermite-evaluation}. For this application, the number of additions performed by Algorithm~\ref{alg:hermite-evaluation} may be reduced by adapting the algorithm to take into account the zeros that initially occupy the $2^nq-\ell$ rightmost entries of the vector $(a_0,\dotsc,a_{2^nq-1})$.

\section{Interpolation algorithm}\label{sec:interpolation}

To solve the Hermite interpolation problem for arbitrary lengths, we use an approach analogous to that employed by the inverse truncated FFT algorithm of Larrieu~\cite{larrieu2017} by reducing to a length $2^nq$ problem under the assumption that the new entries of the output that result from extending the problem are provided as inputs. Thus, we consider the following problem in this section: given $c\in\{1,\dotsc,2^nq\}$ and $(h_0,\dotsc,h_{c-1},f_c,\dotsc,f_{2^nq-1})\in\F^{2^nq}$, compute $f_0,\dotsc,f_{c-1}\in\F$ such that $F=\sum^{2^nq-1}_{i=0}f_ix^i$ satisfies $(D^{i\bdiv{q}}F)(\omega_{i\bmod{q}})=h_i$ for $i\in\{0,\dotsc,c-1\}$. Here, existence and uniqueness of $f_0,\dotsc,f_{c-1}$ follow readily from Lemma~\ref{lem:existence-uniqueness}. The length $\ell$ Hermite interpolation problem is then captured as an instance of the new problem by taking $n=\ceil{\log_2\ceil{\ell/q}}$, $c=\ell$ and $f_c=\dotsb=f_{2^nq-1}=0$. 

The Hermite interpolation algorithm is described in Algorithm~\ref{alg:hermite-interpolation}. If $c>2^{n-1}q$, then the algorithm closely follows the approach described in Section~\ref{sec:setup} by recursively computing the polynomials $F_0$ and $F_1+D^{2^{n-1}}F_0$, before using Lemma~\ref{lem:derivative} and the expansion~\eqref{eqn:expanded-taylor} to compute the desired coefficients of $F$. The recursive call used to recover $F_1+D^{2^{n-1}}F_0$ cannot be made without first computing the coefficient of $x^i$ in the polynomial for $i\geq c-2^{n-1}q$. Consequently, after the algorithm has recovered $F_0$, the required coefficients are computed by function the \textsf{\RightFunction}, which steps through Lines~\ref{eval:taylor-start} to~\ref{eval:hasse-end} of Algorithm~\ref{alg:hermite-evaluation} while only modifying those entries $a_i$ with indices $i\geq c$. If $c\leq 2^{n-1}q$, then the function \textsf{\LeftFunction} from Algorithm~\ref{alg:hermite-evaluation} is used to recover the coefficient of $x^i$ in $F_0$ for $i\geq c$ before the remaining coefficients of the polynomial are recursively computed. The function \textsf{\LeftFunction}, which is its own inverse for fixed $c$, is then used to compute the lower order coefficients of the output. The base case of the recursion is handled by an algorithm \textsf{Interpolate} that satisfies the specifications of Algorithm~\ref{alg:interpolation}.

\begin{algorithm}[h]
	\caption{$\textsf{Interpolate}((a_0,\dotsc,a_{q-1}),c)$}
	\label{alg:interpolation}
	\begin{algorithmic}[1]
		\Require $(a_0,\dotsc,a_{q-1})\in\F^q$ and $c\in\{1,\dotsc,q\}$ such that for some polynomial $F=\sum^{q-1}_{i=0}f_ix^i\in\F[x]_q$ the following conditions hold:
		\begin{enumerate}
			\item $a_i=F(\omega_i)$ for $i\in\{0,\dotsc,c-1\}$, and
			\item $a_i=f_i$ for $i\in\{c,\dotsc,q-1\}$.
		\end{enumerate}
		\Ensure $a_i=f_i$ for $i\in\{0,\dotsc,q-1\}$.
	\end{algorithmic}
\end{algorithm}

Algorithm~\ref{alg:interpolation} may be realised with a complexity of $\bigO(\mult(q)+\mult(c)\log c)$ operations in $\F$ by the use of a fast Chinese remainder algorithm: $\sum^{c-1}_{i=0}f_ix^i$ is equal to the sum of the polynomial $C\in\F[x]_c$ that satisfies $C(\omega_i)=F(\omega_i)$ for $i\in\{0,\dotsc,c-1\}$, and the remainder of $\sum^{q-1}_{i=c}f_ix^i$ upon division by $\prod^{c-1}_{i=0}(x-\omega_i)$, with the product being computed as part of the Chinese remainder algorithm. If the enumeration of the field may be chosen freely and its multiplicative group has smooth cardinality, then a better complexity is obtained by using the inverse truncated FFT algorithm of Larrieu~\cite{larrieu2017}. In doing so, one should set $\omega_{q-1}=0$ so that only a single addition is required on top of the call to the FFT algorithm, since $F(\omega)=f_{q-1}+\sum^{q-2}_{i=0}f_i\omega^i$ for nonzero $\omega\in\F$.

\begin{algorithm}[h]
	\caption{$\textsf{HermiteInterpolate}((a_0,\dotsc,a_{2^nq-1}),c)$}
	\label{alg:hermite-interpolation}
	\begin{algorithmic}[1]
		\Require $(a_0,\dotsc,a_{2^nq-1})\in\F^{2^nq}$ and $c\in\{1,\dotsc,2^nq\}$ such that $n\in\N$ and for some polynomial $F=\sum^{2^nq-1}_{i=0}f_ix^i\in\F[x]_{2^nq}$ the following conditions hold:
		\begin{enumerate}
			\item $a_i=(D^{i\bdiv{q}}F)(\omega_{i\bmod{q}})$ for $i\in\{0,\dotsc,c-1\}$, and
			\item $a_i=f_i$ for $i\in\{c,\dotsc,2^nq-1\}$.
		\end{enumerate}
		\Ensure $a_i=f_i$ for $i\in\{0,\dotsc,2^nq-1\}$.
		\If{$n=0$}
			\State \Call{Interpolate}{$(a_0,\dotsc,a_{q-1}),c$}
			\hfill\textit{/* Algorithm~\ref{alg:interpolation} */}
		\ElseIf{$c>2^{n-1}q$}
			\State\label{interp:large-c-left-recursion} 
			\Call{HermiteInterpolate}{$(a_0,\dotsc,a_{2^{n-1}q-1}),2^{n-1}q$}
			\State\label{interp:inverse-complete-right} \Call{\RightFunction}{$(a_0,\dotsc,a_{2^nq-1}),c$}
			\State\label{interp:large-c-right-recursion} \Call{HermiteInterpolate}{$(a_{2^{n-1}q},\dotsc,a_{2^nq-1}),c-2^{n-1}q$}
			\For{$i=q/2,\dotsc,q-1$}\label{interp:hasse-start}
				\For{$j=0,\dotsc,2^{n-1}-1$}
					\State $a_{2^ni+j}\leftarrow a_{2^ni+j}+a_{2^ni+j-(q-1)2^{n-1}}$
				\EndFor
			\EndFor\label{interp:hasse-end}
			\For{$i=2^{n-1},\dotsc,2^{n-1}(q+1)-1$}\label{interp:taylor-start}
				\State $a_i\leftarrow a_i+a_{2^{n-1}(q-1)+i}$
			\EndFor\label{interp:taylor-end}
		\Else
			\State\label{interp:inverse-complete-left} \Call{\LeftFunction}{$a,c$}\hfill\textit{/* From Algorithm~\ref{alg:hermite-evaluation} */}
			\State\label{interp:small-c-left-recursion} \Call{HermiteInterpolate}{$(a_0,\dotsc,a_{2^{n-1}q-1}),c$}
			\State\label{interp:complete-left} \Call{\LeftFunction}{$a,0$}
		\EndIf
		\vspace{-0.65\baselineskip}
		\Statex\hspace*{-\algorithmicindent}\hrulefill
		\Function{\RightFunction}{$(a_0,\dotsc,a_{2^nq-1}),c$}:
			\For{$i=2^{n-1}(q+1)-1,2^{n-1}(q+1)-2,\dotsc,c$}\label{interp:partial-taylor-start}
				\State $a_i\leftarrow a_i + a_{2^{n-1}(q-1)+i}$
			\EndFor\label{interp:partial-taylor-end}
			\State\label{interp:partial-hasse-start} $t\leftarrow c\bdiv{2^n}$, $r\leftarrow\min(c\bmod{2^n},2^{n-1})$
			\For{$j=0,\dotsc,r-1$}
				\For{$i=t+1,\dotsc,q-1$}
					\State $a_{2^ni+j}\leftarrow a_{2^ni+j}+a_{2^ni+j-(q-1)2^{n-1}}$
				\EndFor
			\EndFor
			\For{$j=r,\dotsc,2^{n-1}-1$}
				\For{$i=t,\dotsc,q-1$}
					\State $a_{2^ni+j}\leftarrow a_{2^ni+j}+a_{2^ni+j-(q-1)2^{n-1}}$
				\EndFor
			\EndFor\label{interp:partial-hasse-end}
		\EndFunction
	\end{algorithmic}
\end{algorithm}

\begin{proposition} Algorithm~\ref{alg:hermite-interpolation} is correct if 
Algorithm~\ref{alg:interpolation} is correctly implemented.
\end{proposition}
\begin{proof} Under the assumption that Algorithm~\ref{alg:interpolation} has been correctly implemented, we use induction to show that for all $n\in\N$, Algorithm~\ref{alg:hermite-interpolation} produces the correct output when given inputs $(a_0,\dotsc,a_{2^nq-1})\in\F^{2^nq}$ and $c\in\{1,\dotsc,2^nq\}$. Therefore, suppose that Algorithm~\ref{alg:interpolation} has been correctly implemented. Then for inputs with $n=0$, the algorithm trivially produces the correct output since Algorithm~\ref{alg:interpolation} is simply applied in this case. Let $n\geq 1$ and suppose that Algorithm~\ref{alg:hermite-interpolation} functions correctly for all inputs with smaller values of~$n$. Suppose that $(a_0,\dotsc,a_{2^nq-1})\in\F^{2^nq}$ and $c\in\{1,\dotsc,2^nq\}$ are given to the algorithm as inputs, and let $F\in\F[x]_{2^nq}$ be the corresponding polynomial for which the input requirements are satisfied. Let $F_0,F_1\in\F[x]_{2^{n-1}q}$ such that~\eqref{eqn:taylor-expansion} and, equivalently, \eqref{eqn:expanded-taylor} hold.
	
Suppose that $c>2^{n-1}q$. Then Lemma~\ref{lem:recursion} implies that~\eqref{eqn:evaluation-left} initially holds for $i\in\{0,\dotsc,2^{n-1}q-1\}$. Consequently, the induction hypothesis and Lemma~\ref{lem:existence-uniqueness} imply that the recursive call of Line~\ref{interp:large-c-left-recursion} sets $a_i$ equal to the coefficient of $x^i$ in $F_0$ for $i\in\{0,\dotsc,2^{n-1}q-1\}$. Thus, when the function \textsf{\RightFunction} is called in Line~\ref{interp:inverse-complete-right}, \eqref{eqn:expanded-taylor} implies that Lines~\ref{interp:partial-taylor-start} and~\ref{interp:partial-taylor-end} of the function set $a_{2^{n-1}q+i}$ equal to the coefficient of $x^i$ in $F_1$ for $i\in\{c-2^{n-1}q,\dotsc,2^{n-1}q-1\}$. Then Lemma~\ref{lem:derivative} implies that Lines~\ref{interp:partial-hasse-start} to~\ref{interp:partial-hasse-end} of the function set $a_{2^{n-1}q+i}$ equal to the coefficient of $x^i$ in $F_1+D^{2^{n-1}}F_0$ for $i\in\{c-2^{n-1}q,\dotsc,2^{n-1}q-1\}$. The entries $a_{2^{n-1}q},\dotsc,a_{c-1}$ are so far unchanged by the algorithm. Thus, Lemma~\ref{lem:recursion} implies that~\eqref{eqn:evaluation-right} holds for $i\in\{0,\dotsc,c-2^{n-1}q-1\}$. The induction hypothesis and Lemma~\ref{lem:existence-uniqueness} therefore imply that the recursive call of Line~\ref{interp:large-c-right-recursion} sets $a_{2^{n-1}q+i}$ equal to the coefficient of $x^i$ in $F_1+D^{2^{n-1}}F_0$ for $i\in\{0,\dotsc,c-2^{n-1}q-1\}$. Hence, after the recursive call, the left half of the vector $(a_0,\dotsc,a_{2^nq-1})$ contains the coefficients of $F_0$, while its right half contains the coefficients of $F_1+D^{2^{n-1}}F_0$. Consequently, Lemma~\ref{lem:derivative} implies that Lines~\ref{interp:hasse-start} to~\ref{interp:hasse-end} set $a_{2^{n-1}q+i}$ equal to the coefficient of $x^i$ in $F_1$ for $i\in\{0,\dotsc,2^{n-1}q-1\}$, then \eqref{eqn:expanded-taylor} implies that Lines~\ref{interp:taylor-start} to~\ref{interp:taylor-end} set $a_i$ equal to the coefficient of $x^i$ in $F$ for $i\in\{0,\dotsc,2^nq-1\}$. The algorithm stops at this point, and thus produces the correct output.	
	
Suppose now that $c\leq 2^{n-1}q$. Then~\eqref{eqn:expanded-taylor} implies that the call to \textsf{\LeftFunction} in Line~\ref{interp:inverse-complete-left} sets $a_i$ equal to the coefficient of $x^i$ in $F_0$ for $i\in\{c,\dotsc,2^{n-1}q-1\}$. This call to \textsf{\LeftFunction} does not modify $a_0,\dotsc,a_{c-1}$. Thus, Lemma~\ref{lem:recursion} implies that~\eqref{eqn:evaluation-left} holds for $i\in\{0,\dotsc,c-1\}$ when the recursive call of Line~\ref{interp:small-c-left-recursion} is made. The induction hypothesis and Lemma~\ref{lem:existence-uniqueness} therefore imply that Line~\ref{interp:small-c-left-recursion} sets $a_i$ equal to the coefficient of $x^i$ in $F_0$ for $i\in\{0,\dotsc,c-1\}$. Hence, after the recursive call, the left half of the vector $(a_0,\dotsc,a_{2^nq-1})$ contains the coefficients of $F_0$, while the entries in the right half still retain their initial values, with $a_i$ equal to the coefficient of $x^i$ in $F$ for $i\in\{2^{n-1}q,\dotsc,2^nq-1\}$. Thus, \eqref{eqn:expanded-taylor} implies that the call to $\textsf{\LeftFunction}$ in Line~\ref{interp:complete-left} sets $a_i$ equal to the coefficient of $x^i$ in $F$ for $i\in\{0,\dotsc,2^{n-1}q-1\}$. The algorithm stops at this point, and thus produces the correct output.
\end{proof}

\begin{proposition}\label{prop:interpolation-complexity} For $n\in\N$, define $\A_n,\M_n:\{1,\dotsc,2^nq\}\rightarrow\N$ as follows: $\A_n(c)$ and $\M_n(c)$ are respectively the number of additions and multiplications in~$\F$ performed by Algorithm~\ref{alg:hermite-interpolation} (for some implementation of Algorithm~\ref{alg:interpolation}) when given inputs $(a_0,\dotsc,a_{2^nq-1})\in\F^{2^nq}$ and $c\in\{1,\dotsc,2^nq\}$. Then
\begin{multline*}
	\A_n(c)
	\leq
	\A_0(q)\left(\ceil{c/q}-1\right)
	+\A_0(c\bmodstar{q})\\
	+\left(
		\frac{7}{4}\ceil{\log_2\ceil{c/q}}
		-n
		-\frac{3}{4}
	\right)
	\left(\ceil{c/q}-1\right)q
	+(2^n-1)(2q+1)
\end{multline*}
and $\M_n(c)=\M_0(q)(\ceil{c/q}-1)+\M_0(c\bmodstar{q})$ for $n\in\N$ and $c\in\{1,\dotsc,2^nq\}$.
\end{proposition}
\begin{proof} The proof of the bound on $A_n(c)$ follows along similar lines to 
that of Proposition~\ref{prop:evaluation-complexity}, but with the addition of 
having to bound the number of additions performed in 
Lines~\ref{interp:inverse-complete-right} 
and~\ref{interp:inverse-complete-left} of the algorithm. As no multiplications 
are performed by either of these lines, they may be ignored when proving the 
formula for $\M_n(c)$. In doing so, the proof follows along identical lines to 
that of Proposition~\ref{prop:evaluation-complexity}, and is therefore omitted.
	
Suppose that Algorithm~\ref{alg:hermite-interpolation} has been given inputs $(a_0,\dotsc,a_{2^nq-1})\in\F^{2^nq}$ and $c\in\{1,\dotsc,2^nq\}$ for some nonzero $n\in\N$. Let $t=c\bdiv{2^n}$ and $r=\min(c\bmod{2^n},2^{n-1})$, as defined in Line~\ref{interp:partial-hasse-start} of the function \textsf{\RightFunction}. If $c>2^{n-1}q$, then the call to \textsf{\RightFunction} in Line~\ref{interp:inverse-complete-right} of the Algorithm~\ref{alg:hermite-interpolation} performs
\begin{align*}
	\max(2^{n-1}(q+1)-c,0)+\left(2^{n-1}(q-t)-r\right)
	<2^{n-1}+\left(2^{n-1}q-c/2\right)
\end{align*}
additions. In particular, no additions are performed if $c=2^nq$. Consequently, $\A_n$ once again satisfies the recurrence~\eqref{eqn:full-length-recurrence} for nonzero $n\in\N$, and, as a result, also satisfies~\eqref{eqn:full-length-complexity}. If $c\leq 2^{n-1}q$, then Line~\ref{interp:inverse-complete-left} of Algorithm~\ref{alg:hermite-interpolation} performs at most $(2^{n-1}q-c)+(2^n-2^{n-1})$ additions. By summing the contributions of each line of Algorithm~\ref{alg:hermite-interpolation} in the manner of the proof of Proposition~\ref{prop:evaluation-complexity}, with the two bounds used for the contributions of Lines~\ref{interp:inverse-complete-right} and~\ref{interp:inverse-complete-left},  it follows that
\begin{multline*}
	A_n(c)	
	\leq
	A_{n-1}\left(c-\delta_n(c)2^{n-1}q\right)
	+2^{n-1}\left(2q+1\right)\\
	+\delta_n(c)2^{n-1}\left(\A_0(q)+\frac{3}{4}(n-1)q+\frac{q}{2}\right)
	-\left(1-\frac{\delta_n(c)}{2}\right)c
\end{multline*}
for nonzero $n\in\N$ and $c\in\{1,\dotsc,2^nq\}$, where $\delta_n$ is the indicator function defined in the proof of Proposition~\ref{prop:evaluation-complexity}. Therefore,
\begin{multline*}
	\A_n(c)	
	\leq
	\A_0(q)\left(\ceil{c/q}-1\right)
	+\A_0\left(c\bmodstar{q}\right)
	+(2^n-1)(2q+1)\\
	+\left(
		\frac{3}{4}\ceil{\log_2\ceil{c/q}}
		-\frac{1}{4}
	\right)
	\left(\ceil{c/q}-1\right)q
	-\sum^{n-1}_{k=0}
	\left(1-\frac{i_k}{2}\right)
	\Bigg(c-q\sum^{n-1}_{j=k+1}2^ji_j\Bigg)
\end{multline*}
for $n\in\N$ and $c\in\{1,\dotsc,2^nq\}$, where $i_0,\dotsc,i_{n-1}\in\{0,1\}$ are the coefficients of the binary expansion~\eqref{eqn:binary-expansion}. The upper bound on $\A_n(c)$ stated in the proposition is then obtained by observing that
\begin{align*}
	\sum^{n-1}_{k=0}
	\left(1-\frac{i_k}{2}\right)
	\Bigg(c-q\sum^{n-1}_{j=k+1}2^ji_j\Bigg)
	&\geq
	\sum^{n-1}_{k=\ceil{\log_2\ceil{c/q}}}c
	+\sum^{\ceil{\log_2\ceil{c/q}}-1}_{k=\max(\ceil{\log_2\ceil{c/q}}-1,0)}
	\frac{c}{2}\\
	&\geq
	\left(n-\ceil{\log_2\ceil{c/q}}+\frac{1}{2}\right)
	\left(\ceil{c/q}-1\right)q
\end{align*}
for $n\in\N$, $c\in\{1,\dotsc,2^nq\}$ and $i_0,\dotsc,i_{n-1}\in\{0,1\}$ such that~\eqref{eqn:binary-expansion} holds.
\end{proof}

By taking $\A_0$ and $\M_0$ to be in $\bigO(\mult(q)+\mult(c)\log c)$, it follows from Proposition~\ref{prop:interpolation-complexity} that the length $\ell$ Hermite interpolation problem can be solved with $\bigO(\mult(\ell)\log q+\ell\log\ceil{\ell/q})$ operations in $\F$ by Algorithm~\ref{alg:hermite-interpolation}. The number of additions performed by the algorithm in this setting may once again be reduced by taking into account the zeros that initially occupy the $2^nq-\ell$ rightmost entries of the vector $(a_0,\dotsc,a_{2^nq-1})$. Moreover, as some of these entries are changed during the course of the algorithm, but ultimately are equal to zero again at its end, it is possible to save further additions by not performing those steps specific to restoring the entries to zero.

%-----------------------------------------
% Bibliography
%-----------------------------------------
\bibliographystyle{amsplain}

\end{document}